\documentclass{article}
\usepackage{amssymb}
\usepackage{amsfonts}
\usepackage{setspace}

\newtheorem{theorem}{Theorem}

\newtheorem{lemma}[theorem]{Lemma}

\newtheorem{proposition}[theorem]{Proposition}
\newtheorem{remark}[theorem]{Remark}

\newenvironment{proof}[1][Proof]{\noindent\textbf{#1.} }{\ \rule{0.5em}{0.5em}}
\input{tcilatex}
\begin{document}

\begin{center}
{\LARGE Information Percolation: Some General Cases\smallskip \bigskip\
with\medskip\ an Application to Econophysics\bigskip }

BY ALAIN$\ $B\'{E}LANGER$^{\ast }${\large \ }AND{\large \ }GASTON GIROUX$%
^{\dag }${\small \ \bigskip }

\textit{Universit\'{e} de Sherbrooke}
\end{center}

\bigskip 

{\small We describe, at the microscopic level, the dynamics of }$N${\small \
interacting components where the probability is very small when }$N${\small %
\ is large that a given component interact more than once, directly or
indirectly, up to time }$t,${\small \ with any other component. Due to this
fact, we can consider, at the macroscopic level, the quadratic system of
differential equations associated with the interaction and establish an
explicit formula for the solution of this system. We moreover apply our
results to some models of Econophysics.}

\section{Introduction}

In their paper, Duffie-Sun (2012) (see also Duffie (2012)), the authors
provide a mathematical foundation for independent random matching of a large
population. Here we develop an approach, inspired by Kac (1956), where the
random matching is instead asymptotically independent. To do so, we start
with a sequence of dynamical sets of interacting components, one for each
integer $N.$ For these dynamical systems we can show that when $N$ is large
the probability is very small that a component has interacted more than once
directly or indirectly up to time $t$ with any other component. Thanks to
this fundamental property, we can link the microscopic and macroscopic
levels using results from the theory of continuous-time Markov chains.

In section 2, we describe these dynamics with their symmetric interaction
kernels. We consider interactions involving $m$ components, for $m\geq 2,$
and we suppose that the intensities of these dynamics have an adequate
dependence on $N$. Our techniques enable us to obtain an explicit formula
for the associated quadratic system of diffferential equations. We thereby
extend the results first obtained in Duffie-Manso ( 2007) and pursued in
Duffie-Giroux-Manso (2010). We note that our formula is valid for any
interaction kernel and it is more explicit than the one obtained for the
particular kernel considered in the latter article. It enables us, in
particular, to obtain new results for models of Econophysics.\bigskip 

\noindent $\overline{%
\begin{array}{l}
\text{{\small Dated 15 February 2012.}} \\ 
^{\ast }\text{ {\small D\'{e}partement de finance, Universit\'{e} de
Sherbrooke,}}%
\end{array}%
}$

\noindent {\small Sherbrooke, Canada, J1K 2R1. E-mail:
alain.a.belanger@usherbrooke.ca}

\noindent $^{\dag }${\small \ 410 rue de Vimy, apt. 1, Sherbrooke, Canada,
J1J 3M9. E-mail: gasgiroux@hotmail.com\bigskip }

{\small \noindent }\noindent {\small AMS classifications:\ 60G55, 34A34,
82C31.\medskip }

{\small \noindent }\noindent {\small Keywords: Large interacting sets,
market equilibrium, Ordinary Differential Equations, continuous-time Markov
chains, Econophysics.}

\newpage 

The statement and the proof that our formula solves the system of
differential equations are done in section 3. In section 4, we present our
applications to Econophysics. While in section 5, we come back to the
fundamental property of the system's dynamics and obtain several
intermediary results leading to its solution.

\section{\protect\bigskip The dynamics.}

We suppose that all components take their values in a measurable space, $(E,%
\mathcal{E}),$ (one can think of ($%
\mathbb{R}
^{d},B(%
\mathbb{R}
^{d}))$ and their interactions are given by a symmetric probability kernel $%
Q $ on the product space ($E^{m},\mathcal{E}^{\otimes m}$ ) for $m\geq 2.$
That is: the function $Q(x_{1},x_{2},...x_{m};C_{1}\times \cdot \cdot \cdot
\times C_{m})$ is measurable in $(x_{1},x_{2},...x_{m});$ is a probability
measure in ($C_{1}\times \cdot \cdot \cdot \times C_{m});$ and satisfies $%
Q(x_{1},x_{2},...x_{m};C_{1}\times \cdot \cdot \cdot \times
C_{m})=Q(x_{\sigma (1)},x_{\sigma (2)},...x_{\sigma (m)};C_{\sigma
(1)}\times C_{\sigma (2)}\times \cdot \cdot \cdot \times C_{\sigma (m)})$
for any permutation $\sigma $ of $\{1,2,...,m\}.$

For each integer $N$, we consider an interacting set of $N$ components which
interact by groups of $m$ according to the kernel $Q.$ The interactions
occur at each jump of a Poisson process with intensity $\lambda \frac{N}{m}$%
. \ Groups are undistinguishable so each group has a probability of $\binom{N%
}{m}^{-1}$ of being involved in a given interaction.

We show, in section 5, the fundamental property that enables us to obtain
the system's solution. In the next section, we describe the formula and show
that it is indeed the system's\ solution.

\section{\protect\bigskip The solution of the macroscopic system.}

The kernel $Q$ allows us to describe the macroscopic evolution of the system
with an associated system of quadratic differential equations via the
evolution of the law of a component. This probability law, denoted $\mu
_{t}, $ evolves with time and is in fact the solution of the Cauchy problem:%
\[
\frac{d\mu _{t}}{dt}=\mu _{t}^{\circ _{m}}-\mu _{t}\;;\mu _{0}=\mu 
\]

where

\bigskip 
\[
\mu ^{\circ _{m}}(C)\triangleq \dint\limits_{%
\mathbb{R}
^{m}}\mu (dx_{1})\mu (dx_{2})...\mu (dx_{m})Q(x_{1},x_{2},...x_{m};C\times
E^{m-1})\text{ for }C\in \mathcal{E}. 
\]

$\ $ The probability law $\mu ^{\circ _{m}}$is the law of a component after
the interaction of $m$ i.i.d. components with law $\mu .$ We can think of it
as the law at the root of the $m$-ary tree with only one interaction. We
will look at all the trees representing the interaction history of a
component up to time $t$. So for a tree, $A$, with more than one
interaction, we divide the tree in $m$ subtrees at that last interaction and
continue recursively up to time 0 to define $\mu ^{\circ _{m}A}$ . (Please
see figure 1 of section 5 for an example of an interaction tree.) Let $%
\mathbb{A}_{n}$ be the set of all trees with $n$ interactions (a.k.a.
nodes), each node producing $m$ branches. If $A_{n}\in \mathbb{A}_{n}$, then 
$\mu ^{\circ _{m}A_{n}}$ denotes the law obtained by iteration of $\mu
^{\circ _{m}}$ through the successsive node of the tree when we place the
law $\mu $ on each leaf of $A_{n}.$

Now we will show that our Cauchy problem has a unique solution which can be
expressed, by conditioning on the number of interactions up to time $t,$ and
then by the component's history. Such conditionings give us

\[
\mu _{t}=\dsum\limits_{n\geq 0}p_{n}(t)\frac{1}{\#_{m}(n)}%
\dsum\limits_{A_{n}\in \mathbb{A}_{n}}\mu ^{\circ _{m}A_{n}}\text{ \ \ \ \ \
\ }(1) 
\]

where $\#_{m}(n)$ $=\dprod\limits_{k=1}^{n-1}((m-1)k+1)$ is the number of
trees with $n$ nodes, taking into account their branching orders; and $%
p_{n}(t)$ is the probability of having $n$ branchings up to time $t.$

Finally, in order to show that the countable convex sum $(1)$ gives the
solution, we need the following lemma which will be proved in section 5.

\begin{lemma}
$p_{n}(t)=\frac{\#_{m}(n)}{(m-1)^{n}n!}e^{-t}(1-e^{-(m-1)t})^{n}.$
\end{lemma}

\begin{theorem}
The convex combination, 
\[
\mu _{t}=\dsum\limits_{n\geq 0}p_{n}(t)\frac{1}{\#_{m}(n)}%
\dsum\limits_{A_{n}\in \mathbb{A}_{n}}\mu ^{\circ _{m}A_{n}} 
\]%
is the solution of the Cauchy problem%
\[
\frac{d\mu _{t}}{dt}=\mu _{t}^{\circ _{m}}-\mu _{t}\;;\mu _{0}=\mu . 
\]
\end{theorem}

\begin{proof}
Since the countable convex sum $(\ast )$ is normally summable, we can
differentiate $\mu _{t}$ term by term to obtain:%
\[
-\mu _{t}+e^{-mt}\dsum\limits_{n\geq 1}n(1-e^{-(m-1)t})^{n-1}\frac{1}{%
(m-1)^{n-1}(n-1)!}\dsum\limits_{A_{n}\in \mathbb{A}_{n}}\mu ^{\circ
_{m}A_{n}} 
\]%
Thus we need to show that: 
\[
\mu _{t}^{\circ _{m}}(C)=e^{-mt}\dsum\limits_{n\geq 0}n(1-e^{-(m-1)t})^{n-1}%
\frac{1}{(m-1)^{n-1}n!}\dsum\limits_{A_{n+1}\in \mathbb{A}_{n+1}}\mu ^{\circ
_{m}A_{n+1}}(C)\text{\ \ \ \ \ }(2) 
\]%
Starting with the definition (on page 3), we have that the LHS of $(2)$ is
equal to 
\begin{eqnarray*}
&&\dint\limits_{%
\mathbb{R}
^{m}}\left( \dsum\limits_{i_{1}\geq 0}e^{-t}(1-e^{-(m-1)t})^{i_{1}}\frac{1}{%
(m-1)^{i_{1}}i_{1}!}\dsum\limits_{A_{i_{1}}\in \mathbb{A}_{i_{1}}}\mu
^{\circ _{m}A_{i_{1}}}(dx_{1})\right) ... \\
&&...\left( \dsum\limits_{i_{m}\geq 0}e^{-t}(1-e^{-(m-1)t})^{i_{m}}\frac{1}{%
(m-1)^{i_{m}}i_{m}!}\dsum\limits_{A_{i_{m}}\in \mathbb{A}_{i_{m}}}\mu
^{\circ _{m}A_{i_{m}}}(dx_{m})\right) ... \\
&&\;\ \ \ \ \ \ \ \ \ \ \ \ \ \ \ \ \ \ \ \ \ \ \ \ \ \ \ \ \ \ \ \ \ \ \ \
\ \ \ \ \ \ \ \ \ \ \ \ \ \ \ \ \ \ \ \ \ \ \ \ \
...Q(x_{1},...,x_{m};C\times E^{m-1})
\end{eqnarray*}%
\bigskip which is equal to%
\begin{eqnarray*}
&&\dint\limits_{%
\mathbb{R}
^{m}}e^{-mt}\left\{ \dsum\limits_{n\geq
0}(1-e^{-(m-1)t})^{n}\dsum\limits_{i_{1}+...+i_{m}=n}\frac{1}{%
(m-1)^{n}i_{1}!...i_{m}!}\right. ... \\
&&\left. \left( \dsum\limits_{A_{i_{1}}\in \mathbb{A}_{i_{1}}}\mu ^{\circ
_{m}A_{i_{1}}}(dx_{1})\right) ...\left( \dsum\limits_{A_{i_{m}}\in \mathbb{A}%
_{i_{m}}}\mu ^{\circ _{m}A_{i_{m}}}(dx_{m})\right) \right\}
Q(x_{1},...,x_{m};C\times E^{m-1})
\end{eqnarray*}%
which in turn is equal to%
\[
\dint\limits_{%
\mathbb{R}
^{m}}e^{-mt}\left( \dsum\limits_{n\geq 0}(1-e^{-(m-1)t})^{n}\frac{1}{%
(m-1)^{n}n!}F(i_{1},...,i_{m},n,\mu ,A_{i_{1}},....,A_{i_{m},}Q,C)\right) 
\]%
where 
\[
F(i_{1},...,i_{m},n,\mu ,A_{i_{1}},....,A_{i_{m},}Q,C)=... 
\]%
\begin{eqnarray*}
&&\dsum\limits_{i_{1}+...+i_{m}=n}\binom{n}{i_{1}}\binom{n-i_{1}}{i_{2}}....%
\binom{i_{m-1}+i_{m}}{i_{m-1}}\left( \dsum\limits_{A_{i_{1}}\in \mathbb{A}%
_{i_{1}}}\mu ^{\circ _{m}A_{i_{1}}}(dx_{1})\right) ... \\
&&...\left( \dsum\limits_{A_{i_{m}}\in \mathbb{A}_{i_{m}}}\mu ^{\circ
_{m}A_{i_{m}}}(dx_{m})\right) Q(x_{1},...,x_{m};C\times E^{m-1})
\end{eqnarray*}%
And this last expression is a decomposition of the trees $A_{n+1}\in \mathbb{%
A}_{n+1}$ appearing in the RHS of $(2)$ in $m$ subtrees after the first node
(taking the branching order into account). The two expressions are therefore
equal and this proves the theorem.
\end{proof}

\begin{remark}
This section brings a simplification to the special case studied in section
3 of Duffie-Giroux-Manso (2010) where the term $a_{(m-1)(n-1)+1}$ can now be
given explicitly as $\frac{1}{(m-1)^{n-1}(n-1)!}$. It also suggests that
some of the results of Duffie-Malamud-Manso (2009) can be extended to the
case where information exchanges involve $m$ agents.
\end{remark}

\begin{remark}
We call the law $\mu _{t}=\dsum\limits_{n\geq 0}e^{-t}(1-e^{-(m-1)t})^{n}%
\frac{1}{(m-1)^{n}n!}\dsum\limits_{A_{n}\in \mathbb{A}_{n}}\mu ^{\circ
_{m}A_{n}}$ an extended Wild sum [11] and note that the convex combination
we obtain for the case $m=2$ is indeed the Wild sum, $\mu
_{t}=\dsum\limits_{n\geq 0}e^{-t}(1-e^{-t})^{n}\frac{1}{n!}%
\dsum\limits_{A_{n}\in \mathbb{A}_{n}}\mu ^{\circ _{m}A_{n}}$, now
well-known in the statistical physics of gases since the work of Kac (1956)
[6].
\end{remark}

\section{An application in Econophysics: kinetic models with random
perturbations}

\bigskip In Ferland-Giroux (1991) the authors study a class of kinetic
equations of Kac's type and they show, for binary collisions, a convergence
to the invariant law at an exponential rate. Bassetti \textit{et al }(2011)
show a similar result using different methods. The convergence in
Ferland-Giroux (1991) is obtained along a set of convenient test functions
with a telescoping technique due to Trotter (1959) and with the use of a
version of Wild sums obtained from judicious conditioning. There is
therefore the possibility of extending the results of
Bassetti-Ladelli-Toscani (2011) for $m\geq 2.$ We do this extension here
only in the context of Ferland-Giroux (1991). For a review of Econophysics,
in the case $m=2,$ one can read, among others, D\"{u}ring-Matthes-Toscani
(2010) \ and Bassetti \textit{et al } (2011) for more recent work.

Let $\mu $ denote a probability law on $%
\mathbb{R}
$ and let $S_{m}(\mu )$ denote the law of $%
H_{1}X_{1}+H_{2}X_{2}+...+H_{m}X_{m}$ where the random variables $\{X_{i}\}$
are independent and of law $\mu $ and the variables $\{H_{i}\}$ are
independant of each other and of the $X$ variables. Under very general
conditions, R\"{o}sler (1992) has shown that the transfomation $S_{m}$ has a
fixed point. We will suppose here that the $H$-variables have values in $%
[0,1]$ , that their mean is $\frac{1}{m}$ and that they are not Bernoulli.
We then have $E[H_{1}+H_{2}+...+H_{m}]=1$ and $S_{m}(\mu )$ has the same
first moment as $\mu .$ If we suppose moreover that $%
E[H_{1}^{2}+H_{2}^{2}+...+H_{m}^{2}]<1$ then R\"{o}sler(1992) gives us the
existence of a fixed point, denoted $\gamma $.This fixed point has a second
moment as soon as $\mu $ does. Our goal is to establish the following result.

\begin{theorem}
If we suppose that $\mu $ has a finite second moment then we have that the
law $\mu _{t}=\dsum\limits_{n\geq 0}e^{-t}(1-e^{-(m-1)t})^{n}\frac{1}{%
(m-1)^{n}n!}\dsum\limits_{A_{n}\in \mathbb{A}_{n}}\mu ^{\circ _{m}A_{n}}$
converges to $\gamma $ at the exponential rate $\eta
=1-E[H_{1}^{2}+H_{2}^{2}+...+H_{m}^{2}].$
\end{theorem}

\noindent Our proof is an extension of Ferland-Giroux (1991) which treats
the case $m=2.$

\begin{proof}
Let us first consider the tree with $m$ leaves, denoted $A_{1}.$ On each one
of its leaves put independent random variables of law $\mu $. Assume that
these variables interact at a node to give $%
H_{1}X_{1}+H_{2}X_{2}+...+H_{m}X_{m}.$ Let us call $\mu ^{\circ _{m}A_{1}}$
, or more simply $\mu ^{\circ _{m}},$ the law of this variable. In a similar
fashion, we can consider $\gamma ^{\circ _{m}}$(which is $\gamma $ since it
is a fixed point)$.$ We will first consider the differences $|<\mu ^{\circ
_{m}},f>-<\gamma ^{\circ _{m}},f>|$ for each $f\in C_{b}^{2}.$ One \ way to
bound this difference is to use the telescoping technique of Trotter (1959)
where we replace one by one (from the left say) the variables with law $\mu $
by variables with law $\gamma $. We then obtain a sum of $m$ terms of the
form $|<\gamma ^{\circ _{k}}\circ \mu ^{\circ _{m-k}},f>-<\gamma ^{\circ
_{k+1}}\circ \mu ^{\circ _{m-k-1}},f>|$ which we will bound. In our
particular models we can write down these expressions \ explicitly as%
\begin{eqnarray*}
&&|E\left[
f(H_{1}Y_{1}+H_{2}Y_{2}+...+H_{k}Y_{k}+H_{k+1}X_{k+1}+...+H_{m}X_{m})\right]
- \\
&&E\left[
f(H_{1}Y_{1}+H_{2}Y_{2}+...+H_{k}Y_{k}+H_{k+1}Y_{k+1}+H_{k+2}X_{k+2}+...+H_{m}X_{m})%
\right] |
\end{eqnarray*}%
where $Y_{i}:i=1,...k+1$ have law $\gamma $ and $X_{j}$, with $j=k+1,...,m,$
have law $\mu .$ Let $%
R_{k}=H_{1}Y_{1}+H_{2}Y_{2}+...+H_{k}Y_{k}+H_{k+2}X_{k+2}+...+H_{m}X_{m}.$
Then we have 
\[
f(R_{k}+H_{k+1}X_{k+1})=f(R_{k})+f^{\prime }(R_{k})H_{k+1}X_{k+1}+f^{\prime
\prime }(R_{k}^{\ast })\left( H_{k+1}X_{k+1}\right) ^{2} 
\]%
and 
\[
f(R_{k}+H_{k+1}Y_{k+1})=f(R_{k})+f^{\prime }(R_{k})H_{k+1}Y_{k+1}+f^{\prime
\prime }(R_{k}^{\ast \ast })\left( H_{k+1}Y_{k+1}\right) ^{2}. 
\]%
Which in turn give us
\end{proof}

\[
|<\mu ^{\circ _{m}},f>-<\gamma ^{\circ _{m}},f>|\leq cE\left[
\dsum\limits_{i=1}^{m}H_{i}^{2}\right] . 
\]%
We now need to iterate the process according to the different trees. Let $%
\widehat{A}_{n}$ denote the set of leaves of the tree $A_{n}.$ Then the
contribution of a leaf $u\in \widehat{A}_{n}$ through its interactions down
to the bottom of the tree is a product of the variables $H_{i}$ with $%
i=1,...,m.$ Let us denote this product by $C_{u}.$ For the tree $A_{n}$, the
result of its interactions through the bottom of the tree will therefore go
from $\dsum\limits_{u\in \widehat{A}_{n}}C_{u}X_{u},$ when all the variables
put on leaves have law $\mu ,$ to $\dsum\limits_{u\in \widehat{A}%
_{n}}C_{u}Y_{u}$ when all the variables have law $\gamma .$ Applying the
same techniques as above, namely a (longer) telescoping and a Taylor series
expansion, we get 
\[
|<\mu ^{\circ _{m}A_{n}},f>-<\gamma ^{\circ _{m}},f>|\leq cE\left[
\dsum\limits_{u\in \widehat{A}_{n}}C_{u}^{2}\right] 
\]%
Let $e_{n}=\frac{1}{(m-1)^{n}n!}\dsum\limits_{A_{n}\in \mathbb{A}_{n}}E\left[
\dsum\limits_{u\in \widehat{A}_{n}}C_{u}^{2}\right] .$ If we decompose $%
A_{n} $ in its $m$ subtrees from its first node we can apply a similar
reasoning to Ferland-Giroux (1991) in order to obtain that $e_{n}\leq
cn^{a-1}$ where $a=\frac{1-\eta }{m-1}$ and $\eta $ is such that $1-\eta =E%
\left[ \dsum\limits_{i=1}^{m}H_{i}^{2}\right] .$

We now reformulate these assertions in propositions and and give their proof.

\begin{lemma}
We have $e_{n}=\left( \frac{1}{n}\right) \left( \frac{1-\eta }{m-1}\right)
(1+e_{1}+....+e_{n-1}).$ Therefore $e_{n}\leq cn^{a-1}$ with $a=\frac{1-\eta 
}{m-1}.$
\end{lemma}

\begin{proof}
The decomposition of each tree $A_{n}$ in $m$ subtrees $\left\{
A_{n}^{i}\right\} _{i=1}^{m}$ at the first node enables us to write $e_{n}$
as a sum of $m$ similar terms%
\[
f_{i}=\frac{1}{(m-1)^{n}n!}\dsum\limits_{A_{n}\in \mathbb{A}_{n}}E\left[
\dsum\limits_{u\in \widehat{A^{i}}_{n}}C_{u}^{2}\right] 
\]
where $\widehat{A^{i}}_{n}$ is the set of leaves of the $i^{th}$ subtree. It
suffices to treat the case $i=1.$ Let us decompose $\mathbb{A}_{n},$ the set
of trees with $n$ nodes, by the number, $k$, of nodes of the subtree $%
A_{n}^{1}.$ There are $\binom{n-1}{k}$ $(m-1)^{n-1-k}(n-1-k)!$ such trees.
Indeed, since we need to take into account the order of appearance of these
nodes we have $\binom{n-1}{k}$ choices for the appearances of $A_{n}^{1}$'s
nodes. Then we have $(m-1)^{n-1-k}$ ways to divide the remaining nodes in
the other $m-1$ trees and finally, there are $(n-1-k)!$ choices for the
nodes' appearances. Note that this number simplifies to $\frac{(n-1)!}{k!}%
(m-1)^{n-1-k}.$ If we denote by $\mathbb{A}_{n,k}$ the subset of $\mathbb{A}%
_{n}$ formed by the trees $A_{n}$ for which their subtree $A_{n}^{1}$ has $k$
nodes, we can then express $f_{i}$ as 
\begin{eqnarray*}
f_{i} &=&\frac{1}{(m-1)^{n}n!}\dsum\limits_{k=0}^{n-1}\dsum\limits_{A_{n}\in 
\mathbb{A}_{n,k}}E\left[ \dsum\limits_{u\in \widehat{A^{i}}_{n}}C_{u}^{2}%
\right] \\
&=&\frac{E\left[ H_{i}^{2}\right] }{(m-1)n}\dsum\limits_{k=0}^{n-1}\frac{1}{%
(m-1)^{k}k!}\dsum\limits_{A_{k}\in \mathbb{A}_{n,k}}E\left[
\dsum\limits_{u\in \widehat{A^{i}}_{k}}C_{u}^{2}\right] \\
&=&\frac{E\left[ H_{i}^{2}\right] }{(m-1)n}\dsum\limits_{k=0}^{n-1}e_{k}%
\text{ , with }e_{0}=1.
\end{eqnarray*}%
This proves the first assertion. A similar calculation to the proof of lemma
3 in Ferland-Giroux (1991) gives us the second result.
\end{proof}

\begin{proposition}
For all $f\in C_{b}^{2}$ we have $|<\mu _{t},f>-<\gamma ,f>|\leq ce^{-\eta
t} $
\end{proposition}

\begin{proof}
Once again, it suffices to follow the proof of theorem 3 in Ferland-Giroux
(1991) with the extended Wild sum and replacing $1-\eta $ by $\frac{1-\eta }{%
m-1}.$
\end{proof}

\section{\protect\bigskip The extended Wild sums.}

\subsection{\protect\bigskip When graphs become trees}

In all our cases, we have an underlying market structure which is a Kac walk
with interactions involving $m$ agents$.$ We add exponential times to obtain
a marked Poisson process whose marks are horizontal lines linking the agents
participating in a given interaction. This enables us to describe the limit
law of an agent, under an appropriate conditioning, as a countable convex
combination on trees which is, as we have shown in section 3, the global
solution of the associated differential equation on the space of probability
laws.

Here we explain how we came to that convex combination. We start our study
by an analysis of the dynamics of the intrinsic structure of the large set
of interacting agents when the number of agents increases. We assume that
each interaction involves $m$ agents, $m\geq 2$. More specifically, we
consider a set of $N$ agents whose interactions happen at unexpected times
so these interactions' occurrences follow a Poisson process$.$ Since agents
are interchangeable, each group has an equal probability of meeting of $%
\left( 
\begin{array}{c}
N \\ 
m%
\end{array}%
\right) ^{-1}.$ If we suppose the intensity of the meetings to be $\frac{N}{m%
}$ then each agent has a meeting rate $\lambda $ which can be assumed to
equal $1$ under a time change.

For $N$ fixed and starting at time $0$, we assign a vertical position to
each agent. The down movement represents the passage of time, see figure 1
on page 9. Each time a group of agents interacts, we draw a horizontal line
between those agents and we draw a vertical line at each agent's position
connecting $0$ to the horizontal line just drawn, so we see a random graph
being formed. When we stop this graph at time $t$, we obtain the finite
graph of all interactions that have taken place. Moreover, the history up to
time $t$ of a given agent, call it $P$, is described by the random sub-graph
connecting all agents who have interacted directly or indirectly with $P$. A
sample history of $P$'s meetings/interactions may look like figure 1 below.

\FRAME{ftbpF}{4.9078in}{6.832in}{0in}{}{}{dessinotctree.gif}{\special%
{language "Scientific Word";type "GRAPHIC";maintain-aspect-ratio
TRUE;display "USEDEF";valid_file "F";width 4.9078in;height 6.832in;depth
0in;original-width 8.7502in;original-height 12.1982in;cropleft "0";croptop
"1";cropright "1";cropbottom "0";filename
'DessinOTCtree.gif';file-properties "XNPEU";}}

The number of meetings is random but we can condition on it. The law of the
finite graph is reversible since the meeting times are uniform on $[0,t]$.
We want to show that a random graph representing the history of $P$ can be
replaced by a random tree as the number of agents, $N$, grows. If we look at
figure 1, we see that the inclusion in the second meeting of one of the
investors having participated in the first one, or the inclusion in the
third meeting of an investor from the first or second one would create a
cycle in our subgraph. As $N$ grows though, the chance of meeting an
investor previously encountered directly or indirectly tends to zero.

To see this, let us consider the subgraph of $P$'s history up to time $t$.
Starting at time $t,$ \ we pursue each one of the encountered vertical lines
in $P$'s history backward in time until we reach the next horizontal line.
If the inclusion of the horizontal line in our graph does not create a cycle
(i.e. no pair of investors were involved directly or indirectly in a
previous meeting) we include the line, if not we remove it. Proceeding in
this fashion up to time $0$ we get a tree with $n$ branchings, say, which
has the same law as the law of a tree obtained by a pure-birth process.
Namely, the tree starting at $P^{\prime }$s vertical line at time $t$ with
intensity $1$ and which at time $0$ has intensity $(m-1)n+1$ and that same
number of leaves. Between two branchings of this process a graph
representing $P$'s meeting history can have a random number of additional
horizontal lines following a Poisson law of parameter at most $\frac{N}{m}%
\left( \left( 
\begin{array}{c}
(m-1)n+1 \\ 
2%
\end{array}%
\right) \left( 
\begin{array}{c}
N \\ 
m%
\end{array}%
\right) ^{-1}\right) $. We will now bound the expectation of these
supplementary horizontal lines by a majorant which tends to $0$ as $N$
increases. Indeed, since the mean number of redundant lines when there are $%
n $ branchings up to time $t$ is at most $\frac{N}{m}\left( \left( 
\begin{array}{c}
(m-1)n+1 \\ 
2%
\end{array}%
\right) \left( 
\begin{array}{c}
N \\ 
m%
\end{array}%
\right) ^{-1}\right) ,$ we have that the mean number of redundant horizontal
lines is bounded above by $\dsum\limits_{n\geq 0}\frac{N}{m}\left( \left( 
\begin{array}{c}
(m-1)n+1 \\ 
2%
\end{array}%
\right) \left( 
\begin{array}{c}
N \\ 
m%
\end{array}%
\right) ^{-1}\right) p_{N,n}(t)$, where $p_{N,n}(t)$ is the probability of
having $n$ branchings up to time $t$ of the pure birth process with
successive branching waiting times following exponential laws of parameter 
\[
\lambda _{N,n}=\frac{N}{m}(\left( m-1)n+1\right) \left( 
\begin{array}{c}
N-((m-1)n+1) \\ 
m-1%
\end{array}%
\right) \left( 
\begin{array}{c}
N \\ 
m%
\end{array}%
\right) ^{-1} 
\]

\bigskip Since 
\begin{eqnarray*}
&&\frac{N}{m}(\left( m-1)n+1\right) \left( 
\begin{array}{c}
N-((m-1)n+1) \\ 
m-1%
\end{array}%
\right) \left( 
\begin{array}{c}
N \\ 
m%
\end{array}%
\right) ^{-1} \\
&=&\frac{((m-1)n+1)\binom{N-((m-1)n+1)}{m-1}}{\binom{N-1}{m-1}}\text{ \ \ \
\ \ \ \ \ \ \ \ }(3) \\
&\leq &(m-1)n+1
\end{eqnarray*}%
then $p_{N,n}(t)$ is stochastically smaller than the law obtained with the
intensities $\lambda _{n}=(m-1)n+1,$ which in turn are less than the
intensities $\overline{\lambda }_{n}=m(n+1).$ Its transition kernel is then
obtained by solving Kolmogorov's affine system of equations:

\begin{eqnarray*}
\frac{d\overline{p}_{t}(0)}{dt} &=&-\overline{p}_{t}(0)\text{ \ \ \ \ \ \ \
\ \ \ \ \ \ \ \ \ \ \ \ \ \ \ \ \ \ \ \ \ \ \ \ \ \ \ \ \ }\ \text{\ \ \ \ \
\ \ \ \ \ \ \ \ \ \ } \\
\frac{d\overline{p}_{t}(n)}{dt} &=&mn\overline{p}_{t}(n-1)-m(n+1)\overline{p}%
_{t}(n)\text{ \ ; }n\geq 1\text{.}
\end{eqnarray*}%
Thus the latter intensities give us a geometric law $\overline{p}%
_{t}(n)=e^{-mt}(1-e^{-mt})^{n}\ $. Since geometric laws have finite moments
of all orders, the mean number of redundant horizontal lines is bounded
above by a quantity converging to $0$.

For more details on Kolmogorov systems of equations for pure birth processes
we refer the reader to Lefebvre (2006), for instance.

Thus, after having specified the initial agents' states and their
interaction kernels, we can approximate $P^{\prime }s$ law using the tree
obtained from removing all redundant horizontal lines from its graph. We
will use this fact in the next sub-section.

\subsection{Limit countable convex combination}

\bigskip We will now show that these random trees whose branching
intensities depend on $N$ can be approximated by trees with branching
intensities independent of $N$. Taking into account that $P$'s tree history
is random with intensities depending on $N,$ we could write $P$'s law,
denoted by $\mu _{t}^{\ast ,N},$ with complex formulae depending on $N.$
Since our markets have a large number of investors, it is preferrable
instead to work with the limit of these laws. We note from $(3)$ above that $%
\lambda _{N,n}\rightarrow ((m-1)n+1).$

\bigskip Let $p_{n}(t)$ $(\triangleq p_{t}(n))$ be the solution of the
affine Kolmogorov system of equations: 
\begin{eqnarray*}
\frac{dp_{t}(0)}{dt} &=&-p_{t}(0)\text{ \ \ \ \ \ \ \ \ \ \ \ \ \ \ \ \ \ \
\ \ \ \ \ \ \ \ \ \ \ \ \ }(4) \\
\frac{dp_{t}(n)}{dt} &=&((m-1)(n-1)+1)p_{t}(n-1)-((m-1)n+1)p_{t}(n)\text{ \
; }n\geq 1\text{.}
\end{eqnarray*}

Recall fron the first section that $\mu _{t}=\dsum\limits_{n\geq 0}p_{n}(t)%
\frac{1}{\#_{m}(n)}\dsum\limits_{A_{n}\in \mathbb{A}_{n}}\mu ^{\circ
_{m}A_{n}}$.

\begin{proposition}
The sequence of \ laws $\mu _{t}^{\ast ,N}$ converges to $\mu _{t}$ as $N$
increases.
\end{proposition}

\bigskip

\begin{proof}
By Kurtz (1969) we have that $p_{N,n}(t)\rightarrow p_{n}(t)$ as $N$
increases. But $(p_{n}(t))_{n\geq 0}$ is a probability law, so for $\epsilon
>0,$ there exists $n(\epsilon )$ such that

$\dsum\limits_{n\geq n(\epsilon )}p_{n}(t)<\epsilon .$ Now let $N(\epsilon )$
be such that $N>N(\epsilon )$ implies that $|p_{N,n}(t)-p_{n}(t)|<\frac{%
\epsilon }{n(\epsilon )}$ for $0\leq n\leq n(\epsilon ).$We then have for $%
C\in \mathcal{E}$ and $N>N(\epsilon )$

\[
|\mu _{t}^{\ast ,N}(C)-\mu _{t}(C)|\leq \dsum\limits_{n=0}^{n(\epsilon
)}|p_{N,n}(t)-p_{n}(t)|+2\epsilon \leq 3\epsilon 
\]%
since $\frac{1}{\#_{m}(n)}\dsum\limits_{A_{n}\in \mathbb{A}_{n}}\mu ^{\circ
_{m}A_{n}}(C)\leq 1$ and ($p_{N,n}(t))_{n\geq 0}$ are probability laws. Our
claim is proved.
\end{proof}

\begin{lemma}
$p_{t}(n)=\frac{\#_{m}(n)}{(m-1)^{n}n!}e^{-t}(1-e^{-(m-1)t})^{n}$
\end{lemma}

\bigskip

\begin{proof}
\bigskip We need to solve the affine Kolmogorov system of equations $(4)$.

Proceeding by induction we have:\bigskip 
\begin{eqnarray*}
\frac{dp_{t}(0)}{dt} &=&e^{-t} \\
\frac{dp_{t}(n)}{dt} &=&((m-1)(n-1)+1)e^{-(n(m-1)+1)t}\dint%
\limits_{0}^{t}e^{(n(m-1)+1)s}p_{s}(n-1)ds\text{ \ }
\end{eqnarray*}%
To prove the lemma it suffices to note that $\#_{m}(n)=%
\#_{m}(n-1)((n-1)(m-1)+1)$ and that $%
e^{(n(m-1)+1)s}e^{-s}(1-e^{-(m-1)s})^{n-1}=e^{(m-1)s}(e^{(m-1)s}-1)^{n-1}$
is the derivative of $\frac{1}{(m-1)n}(e^{(m-1)s}-1)^{n}.$
\end{proof}

And this shows that the limit law of $P$ is indeed the extended\ Wild sum
which we have shown to be the solution of the ODE associated to the
interacting system .

\textbf{Acknowledgements. }The first author would like to thank the Universit%
\'{e} de Sherbrooke and its Facult\'{e} d'administration for their startup
grant.

\section{References}

\begin{enumerate}
\item Bassetti, F., Ladelli, L. and Toscani, G. (2011). Kinetic Models with
randomly perturbed Binary Collisions. J. Stats. Physics, vol. \textbf{142},
no. 4, 686-709.

\item Duffie, D. (2012). Dark Markets: Asset Pricing and Information
Percolation in Over-the-Counter Markets. Princeton Lecture Series.

\item Duffie, D. and Manso, G. (2007). Information Percolation in Large
Markets. American Economic Review, Papers and Proceedings, \textbf{97},
203-209.

\item Duffie, D., Giroux, G. and Manso, G. (2010). Information percolation,
American Economic Journal: Microeconomics, vol. \textbf{2}, 100-111.

\item Duffie, D., Malamud, S. and Manso, G.(2009). Information percolation
with equilibrium search dynamics. Econometrica, vol. \textbf{77}, no. 5,
1513-1574.

\item Duffie, D. and Sun, Y. (2012). The Exact Law of Large Numbers for
Independent Random Matching. Journal of Economic Theory, forthcoming.

\item D\"{u}ring, B., Matthes, D. and Toscani, G. (2008). Kinetic Equations
modeling Wealth Redistribution. Phys. Rev. E \textbf{78}, 056103.

\item Ferland, R. and Giroux, G. (1991). An exponential Rate of Convergence
for a Class of Boltzmann Processes, Stochastics and Stochastics Reports, 
\textbf{35}, 79-91.

\item Kac, M. (1956). Foundations of kinetic theory. Proceedings of the
Third Berkeley Symposium on Mathematical Statistics and Probability,
1954--1955, vol. \textbf{III}, pp. 171--197. University of California Press,
Berkeley and Los Angeles.

\item Kurtz, T. G.(1969). A Note on sequences of continuous parameter Markov
chains. Ann. Math. Statist., \textbf{40}, 1078-1082.

\item Lefebvre, M. (2006). Applied Stochastic Processes. Springer.

\item R\"{o}sler, U. (1992). A fixed point theorem for distributions.
Stochastic Processes and their Applications, \textbf{42}, 195-214.

\item Trotter, H. (1959). An elementary proof of the central limit theorem,
Arch. Math (Basel)\textit{, }\textbf{10}, 226-234.

\item Wild, E (1951). On the Boltzmann equation in the kinetic theory of
gases. Proc. Cambridge Phil. Soc. \textbf{47}, 602-609.\bigskip 

$%
\begin{array}{lllll}
\text{Alain B\'{e}langer,} &  &  &  & \text{Gaston Giroux,} \\ 
\text{D\'{e}partement de finance,} &  &  &  & \text{410 rue de Vimy, apt. 1,}
\\ 
\text{Universit\'{e} de Sherbrooke,} &  &  &  & \text{Sherbrooke, Canada,
J1J 3M9} \\ 
\text{2 500 boul. de l'Universit\'{e},} &  &  &  & \text{E-address:
gasgiroux@hotmail.com} \\ 
\text{Sherbrooke, Canada, J1K 2R1} &  &  &  &  \\ 
\text{E-address: alain.a.belanger@usherbrooke.ca} &  & \text{ \ \ } &  &  \\ 
&  &  &  & 
\end{array}%
$
\end{enumerate}

\end{document}